\documentclass{article}
\usepackage{amsmath,amsfonts,amssymb,amsthm}
\usepackage{paralist}
\usepackage{booktabs}
\usepackage[round]{natbib}
\usepackage{graphicx}

\allowdisplaybreaks
\newcommand{\sumdot}{\text{\tiny$\bullet$}}

\newcommand{\bsa}{\boldsymbol{a}}
\newcommand{\bsb}{\boldsymbol{b}}

\newcommand{\bsx}{\boldsymbol{x}}
\newcommand{\bsy}{\boldsymbol{y}}
\newcommand{\bsz}{\boldsymbol{z}}

\newcommand{\real}{\mathbb{R}}
\newcommand{\natu}{\mathbb{N}}
\newcommand{\tran}{\mathsf{T}}
\newcommand{\dnorm}{\mathcal{N}}

\renewcommand{\ge}{\geqslant}
\renewcommand{\le}{\leqslant}
\renewcommand{\emptyset}{\varnothing}

\newcommand{\ck}{\mathcal{K}}
\newcommand{\cp}{\mathcal{P}}
\newcommand{\ct}{\mathcal{T}}
\newcommand{\cx}{\mathcal{X}}

\newcommand{\law}{\mathcal{L}}

\newcommand{\sphd}{\mathbb{S}^{d-1}}

\newcommand{\e}{\mathbb{E}}
\newcommand{\var}{\mathrm{Var}}
\newcommand{\cov}{\mathrm{Cov}}
\newcommand{\cor}{\mathrm{Corr}}
\newcommand{\cv}{\mathrm{cv}}

\newcommand{\rd}{\,\mathrm{d}}

\newcommand{\one}{\boldsymbol{1}}

\newcommand{\dustd}{\mathbb{U}}

\newcommand{\simiid}{\stackrel{\mathrm{iid}}\sim}

\newcommand{\phm}{\phantom{-}}

\newcommand{\umu}{\underline{\mu}}

\newcommand{\ak}{\mathrm{AK}}

\newtheorem{theorem}{Theorem}
\newtheorem{lemma}{Lemma}
\newtheorem{corollary}{Corollary}

\theoremstyle{definition}
\newtheorem{remark}{Remark}

\author{
Art B. Owen\\Stanford University
\and
Yury Maximov\\ T-4 and CNLS, Los Alamos National Laboratory
\and
Michael Chertkov\\ T-4 and CNLS, Los Alamos National Laboratory 
}
\date{December 2018}
\title{Importance sampling the union of rare events with
an application to power systems analysis}

\begin{document}
\maketitle
\begin{abstract}
We consider importance sampling to estimate the probability $\mu$ of a union of $J$ rare events $H_j$
defined by a random variable $\bsx$.
The sampler we study has been used in spatial statistics, genomics and combinatorics
going back at least to Karp and Luby (1983).
It works by sampling one event at random, then sampling $\bsx$ conditionally
on that event happening and it constructs an unbiased estimate of $\mu$
by multiplying an inverse moment of the number of occuring events by the union bound.
We prove some variance bounds for this sampler. For a sample size of $n$,
it has a variance no larger than $\mu(\bar\mu-\mu)/n$ where $\bar\mu$
is the union bound.
It also has a coefficient of variation no larger than 
$\sqrt{(J+J^{-1}-2)/(4n)}$ regardless of the overlap pattern among the $J$ events.
Our motivating problem comes from power system reliability, where the phase
differences between connected nodes have a joint Gaussian distribution and
the $J$ rare events arise from unacceptably large phase differences.
In the grid reliability problems even some events defined by $5772$ constraints
in $326$ dimensions, with probability below $10^{-22}$,
are estimated with a coefficient of variation of about $0.0024$ with only $n=10{,}000$
sample values.
\end{abstract}

\section{Introduction}

In this paper we consider a mixture importance sampling strategy
to estimate the probability that one or more of a set of rare events takes place.
The sampler repeatedly chooses a rare event at random, 
and then samples the system conditionally on that one event taking place.
For each such sample, the total number of occuring events is recorded
and a certain reciprocal moment of them is used in the estimate.

This method is a special case of an algorithm in~\cite{adle:blan:liu:2008, adle:blan:liu:2012}
for computing exceedance probabilities of Gaussian random fields.
It was used earlier by \cite{shi:sieg:yaki:2007}  and \cite{naim:prie:2001}
for extrema of genomic scan statistics.
\cite{prie:naim:cope:2001} used it for extrema of some spatial statistic
involving marked point processes.  The earliest uses that we know 
are in the computer science literature for enumeration problems like
estimating the cardinality of the union of a given
list of finite sets.  See \cite{karp:luby:1983} and \cite{frig:verc:1985}.
The above cited papers refer to this method as importance sampling.
To distinguish it from other samplers, we will call it ALOE
for ``At Least One rare Event''.  

We develop general bounds for the variance of the 
ALOE importance sampler, and for its coefficient of variation. 
It has a sampling standard deviation
that is no more than some modest multiple of the event
probability.  This is an especially desirable property in
rare event settings. For background on importance sampling
of rare events see \cite{lecu:mand:tuff:2009}.

Our motivating context is the reliability of the electrical grid when
subject to random inputs, such as variable demand by users and
variable production, as occurs at wind farms. 
The rare events describe unacceptably large electrical phase differences at pairs
of connected nodes in the grid.

It is common to use a simplified linear direct current (DC) model of
the electrical grid, because the equations describing alternating current (AC)
are significantly more difficult to work with, and some authors
(e.g., \cite{van2014dc}) find that there is little to be gained from the complexity of an AC model.  
This DC model is presented in \cite{sauer1984active} and \cite{stott2009dc}.
It is also common to model the randomness in the grid
as Gaussian, especially over short time horizons. 

We make both of these simplifications: linearity and Gaussianity.
The probability we consider can then be written
\begin{align}\label{eq:asgaussian}
\mu=\Pr\bigl( \cup_{j=1}^J H_j\bigr),\quad 
%H_j =\{\bsx\in\real^d\mid \bsx^\tran\omega_j \ge\tau_j\},
H_j =\{\bsx^\tran\omega_j \ge\tau_j\},
\quad\text{where $\bsx\sim\dnorm(\eta,\Sigma)$}.
\end{align}
Section~\ref{sec:gaussian} introduces more notation
for problem~\eqref{eq:asgaussian} and develops the ALOE
sampler as an especially convenient version of mixture importance sampling.
In this setting we can compute the union bound 
$\bar\mu=\sum_{j=1}^J\Pr(H_j)\ge\mu.$
%An upper bound for $\mu$ is $\bar\mu=\sum_{j=1}^J\Pr(H_j)$, the union bound. 
Theorem~\ref{thm:muvar} proves that
the ALOE estimate $\hat\mu$
has variance at most $\mu(\bar\mu-\mu)/n$ % \le (J-1)\mu^2/n$ 
when $n$ IID samples are used.   This can be much smaller than $\mu(1-\mu)/n$
which arises from sampling the nominal distribution of $\bsx$.
Section~\ref{sec:isamp} discusses some further sampling properties of our estimator
that hold without the Gaussian assumption.
When there are $J$ events, the variance of $\hat\mu$ is
at most $(J+J^{-1}-2)\mu^2/(4n)$ when the system is sampled $n$ times.
Section~\ref{sec:comparison} compares ALOE
to a state of the art code {\tt mvtnorm}  \citep{mvtnorm}
for estimating the probability that a multivariate
Gaussian of up to $1000$ variables with arbitrary covariance belongs to a given hyperrectangle.
ALOE is simpler and extends to higher dimensions.
When we studied rare event cases, ALOE
was more accurate.  In our examples that are not rare events,
{\tt mvtnorm} was more accurate.
We also make a comparison to a directional sampling method studied
recently by \cite{ahn:kim:2018}.  That method is far better than ALOE
on our low dimensional test problems but very seriously underestimates
the rare event probability on  our high dimensional test problems.
Section~\ref{sec:power} describes the power system application.
Section~\ref{sec:discussion} contains some discussions.
The appendix proves 
Theorem~\ref{thm:muvar} for any set of $J$ events, not just
those given by a Gaussian distribution. The theorem applies
so long as we can sample conditionally on any one event $H_j$
and then determine which other events $H_\ell$ also occur.
We finish this section with some comments and some references.

One common way for rare event sampling to be inaccurate is that
we might fail to obtain any points where the rare event happens. That leads
to a severe under-estimation of the rare event probability. In ALOE, the
corresponding problem is the failure to sample any points where two or more
of the rare constituent events occur. In that case ALOE will return the union
bound as the estimated rare event probability instead of zero.  That is also a setting
where the union bound is likely to be a good approximation.  So ALOE is robust
against severe underestimates of the rare event probability.
The second common problem for rare event sampling is an extreme value of
the likelihood ratio weighting applied to the observations.  
In ALOE, the largest possible
weight is only $J$ times as large as the smallest one.

Our sampler is closely related to instanton methods in
power systems engineering. See 
Chertkov, Pan et al.\ (2011), \nocite{11CPS}
Chertkov, Stepanov et al.\ (2011), \nocite{11CSPB} and \cite{15KHCBB}.  
% Awkward citation here.  
% The first two both come out as Chertkov et al. (2011) so they look the same.
% Fixed by adding second authors manually.
Out of all the configurations of random
inputs to a system, the most probable one causing the failure is
called the instanton. When there are thousands of failure types
there are correspondingly thousands of instantons, each one a conditional mode of the 
distribution of $\bsx$.
Our initial thought was to do importance sampling from a mixture of distributions,
with each mixture component defined by shifting the Gaussian distribution's mean to an instanton.
%What we do instead is sample conditionally on the rare event happening. 
By sampling conditionally on an event,
ALOE avoids wasting samples outside the failure region.  By conditioning instead
of shifting, we get better control over the likelihood ratio in the importance sampler. 
%Each instanton remains the most probable point in its half-space under this sampling. % Clearly true under sampling from its own mixture component; Maybe not, under the full mixture distribution.

ALOE is a form of multiple importance sampling. 
Multiple importance sampling originated in computer graphics \citep{lafo:will:1993,veac:guib:1994}.
\cite{owen:zhou:2000} found a useful way to combine it with control variates
defined by the mixture components.
\cite{elvi:mart:leun:buga:2015,elvi:mart:leun:buga:2015:arxiv} investigate computational efficiency
of some mixture importance sampling and weighting strategies.

We do not consider self-normalized importance sampling  (SNIS) in
this paper. SNIS is useful in settings where we can compute an 
unnormalized version of our target density but cannot sample from 
it efficiently, if at all. SNIS is common in  Bayesian applications
\citep[Chapter 2]{liu:2001}.
For a recent adaptive version of SNIS, see \cite{corn:mari:mira:robe:2012}.
For rare event estimation, we show in the Appendix that
self-normalized importance sampling cannot deliver a coefficient
of variation meaningfully below $2/\sqrt{n}$ asymptotically.
The optimal sampler for SNIS allocates precisely half of its
probability in the rare event and half outside of it.
The optimal plain IS estimator, by contrast,
places all of its probability on the rare event and has zero variance.  
Ordinary importance
sampling can attain much smaller variances, and so we focus on it
for the rare event problem.

\section{Gaussian case}\label{sec:gaussian}
For concreteness, we present ALOE first for Gaussian random variables.
The earliest use we have seen for Gaussian variables is \cite{adle:blan:liu:2008}.
We let $\bsx\in\real^d$ have the standard Gaussian distribution, $\dnorm(0,I)$,
deferring general Gaussians to Section~\ref{sec:generalgaus}.
We are interested in computing the probability 
that $\bsx$ lies outside a polytope $\cp$. 
In our motivating applications, 
the interior of the polytope defines a safe operating 
region and we assume that $\bsx\not\in \cp$ is a rare 
event. 
For $j=1,\dots,J$, define half-spaces
$$H_j = \{\bsx\mid \omega_j^\tran\bsx\ge\tau_j\}$$
where each $\tau_j\in\real$ and $\omega_j\in\real^d$, with $\omega_j^\tran\omega_j=1$. 
Then $\cp = \cap_{j=1}^JH_j^c$ and we want 
to find $\mu=\Pr(\bsx \in H)$ where $H=\cup_{j=1}^JH_j=\cp^c$. 
The set $\cp$ is convex and  not necessarily bounded.
Ordinarily $\tau_j>0$, because we are interested in rare events.

\begin{figure}[t]
\includegraphics[width=.9\hsize]{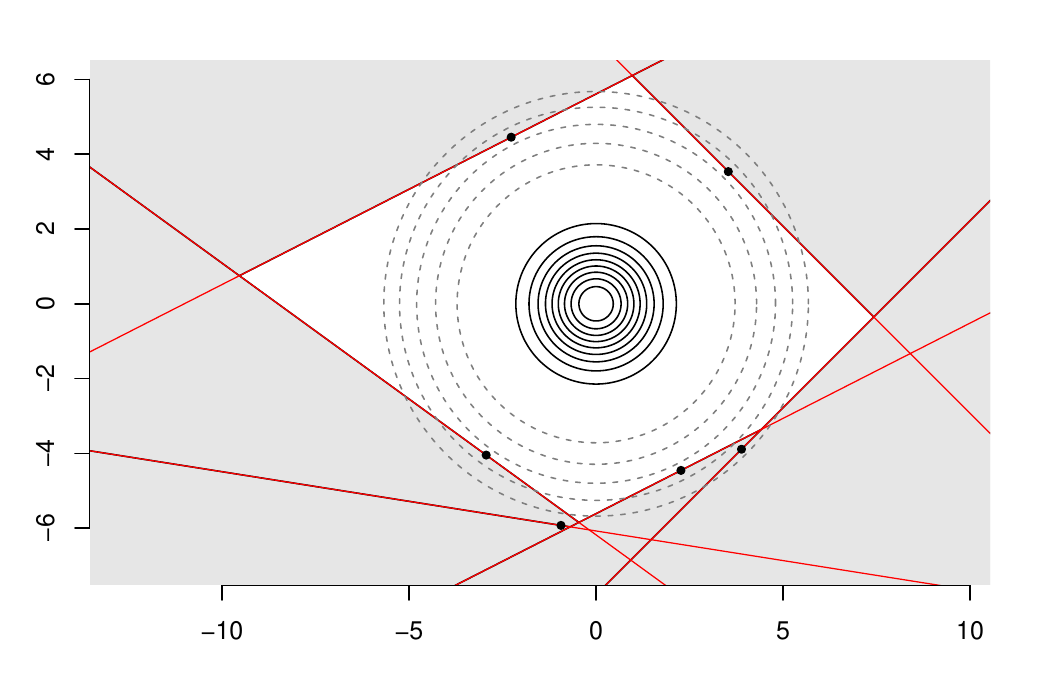}
\caption{\label{fig:gaustope}
The solid circles contain $10$\%, $20$\% up to $90$\% of
the $\dnorm(0,I)$ distribution. The dashed circles contain
all but $10^{-k}$ of the probability for $3\le k\le 7$. The six
solid lines denote half-spaces.  The solid points are the corresponding
conditional modes (instantons). The rare event of interest is $\bsx$
in the shaded region, when $\bsx\sim\dnorm(0,I)$.
}
\end{figure}

The setting is illustrated in Figure~\ref{fig:gaustope} for $J=6$ half-spaces.
In that example, two of the  half-spaces have their
conditional modes inside the union of the other half-spaces. 
One of those  half-spaces is entirely included in the union of the others.

Letting $P_j=\Pr(\bsx\in H_j) =\Phi(-\tau_j)$, we know  that
\begin{align}\label{eq:mubounds}
\max_{1\le j\le J}P_j=:\umu \le \mu \le \bar\mu 
%\equiv\sum_{j=1}^J \Pr(\bsx\in H_j) =\sum_{j=1}^J \Phi(-\tau_j),$$
:=\sum_{j=1}^JP_j.
\end{align}
The right hand side is the union bound which is sometimes
very conservative and sometimes quite accurate.
%Sometimes that bound is very conservative, and other 
%times it is quite accurate. For a given high dimensional
%problem we might not know whether the bound is close.
%Therefore we seek a more accurate estimate of $\mu$. 

We will need to use some inclusion-exclusion formulas, so 
some notation for these follows. 
%For $J\ge1$, let $H_1,\dots,H_J$ be events defined on the  same sample space.  
For any $u\subseteq 1{:}J\equiv \{1,2,\dots,J\}$, 
let $H_u = \cup_{j\in u}H_j$, so $H_j=H_{\{j\}}$ and 
by convention $H_\emptyset=\emptyset$.
We identify the set $H_u$ with the function $H_u(\bsx)=\one\{\bsx\in H_u\}$.
Next define $P_u = \e(H_u(\bsx))$ for $\bsx\sim\dnorm(0,I)$. 
We use $-u$ for complements in $1{:}J$ especially 
within subscripts, and $H_u^c(\bsx)$
for the complementary outcome $1-H_u(\bsx)$. 
%Now $G_u(\bsx)=H_u(\bsx)H_{-u}^c(\bsx)$ describes the 
%event where $\bsx_j\in H_j$ if and only if $j\in u$,
%and $Q_u = \e(G_u(\bsx))$ is the probability that this happens. 
Let $S(\bsx)=\sum_{j=1}^JH_j(\bsx)$ count the number of rare events  that happen. 
For $s=0,1,\dots,J$,  %$R_s = \sum_{u:|u|=s}G_u$
%is the event that $S=s$ and  
let $T_s=\Pr(S=s)$ give the distribution of $S$. 
We use $|u|$ for the cardinality of~$u$.
Our estimand is
\begin{align}\label{eq:incexc}
\mu = P_{1{:}J} = \sum_{|u|>0}(-1)^{|u|-1}P_u,
\end{align}
by inclusion-exclusion.

%\cite{frig:verc:1985} and other papers present the ALOE sampler but do not derive it.  
Here we motivate ALOE as an especially simple mixture sampler.
The mixture components we use are
conditional distributions $q_j =\law(\bsx\mid \omega^\tran_j\bsx\ge\tau_j)$,
for $j=1,\dots,J$.  
%and take $q_0 = p = \dnorm(0,I)$, with $P_0=1$ and $H_0(\bsx)=1$. 
They have probability density functions $q_j(\bsx) = p(\bsx)H_j(\bsx)/P_j$.

Let $\alpha_1,\dots,\alpha_J$ be nonnegative numbers
summing to $1$, and $q_\alpha = \sum_{j=1}^J\alpha_jq_j$.
A mixture importance sampling estimate of $\mu$ based
on $n$ draws $\bsx_i\sim q_\alpha$ is
\begin{align}\label{eq:mualpha}
\hat\mu_\alpha
= \frac1n\sum_{i=1}^n
\frac{ p(\bsx_i)H_{1{:}J}(\bsx_i)  }
{\sum_{j=1}^J\alpha_jq_j(\bsx_i)}
= \frac1n\sum_{i=1}^n
\frac{ H_{1{:}J}(\bsx_i)  }{\sum_{j=1}^J\alpha_jH_j(\bsx_i)P_j^{-1}}.
\end{align}
Notice that $p(\bsx_i)$ has conveniently canceled from numerator and denominator.
%The $\alpha_0$ term represents a defensive mixture component \citep{hest:1995}.
Although the inclusion-exclusion formula~\eqref{eq:incexc}
contains $2^J-1$ nonzero terms,
each summand in the unbiased estimate in~\eqref{eq:mualpha} can be computed
at cost $O(J)$.

% control variate material deprecated

We can induce further cancellation in~\eqref{eq:mualpha}
by making $\alpha_j/P_j$ constant in $j$.
%We discard the defensive component, taking $\alpha_0=0$. 
Taking
$\alpha_j = \alpha_j^* \equiv P_j/\bar\mu$, we get
\begin{align}\label{eq:hatmubal}
%\hat\mu_*\equiv 
\hat\mu_{\alpha^*}=
 \frac{\bar\mu}n\sum_{i=1}^n 
\frac{ H_{1{:}J}(\bsx_i) }{\sum_{j=1}^JH_j(\bsx_i)}
= \frac{\bar\mu}n\sum_{i=1}^n 
%\biggl(\,\sum_{j=1}^JH_j(\bsx_i)\biggr)^{-1},\quad \bsx_i\simiid q_{\alpha^*},
\frac1{S(\bsx_i)},\quad \bsx_i\simiid q_{\alpha^*},
\end{align}
because $H_{1{:}J}(\bsx)=1$ always holds for $\bsx\sim q_{\alpha^*}$.
The estimate~\eqref{eq:hatmubal} is a multiplicative
adjustment to the union bound $\bar\mu$.
The terms $S(\bsx_i)^{-1}$ range from $1$ to $1/J$ and so
we will never get $\hat\mu_{\alpha^*}$ larger than the union bound
or smaller than $\bar\mu/J$.
This is convenient because $\bar\mu \ge\mu\ge\bar\mu/J$
always holds.

\begin{theorem}\label{thm:muvar}
Let $\hat\mu_{\alpha^*}$ be given by~\eqref{eq:hatmubal}.
Then
\begin{align}\label{eq:unbi}
\e(\hat\mu_{\alpha^*})=\mu,
\end{align}
and 
\begin{align}\label{eq:varub}
\var(\hat\mu_{\alpha^*}) 
=\frac1n\biggl( \bar\mu\sum_{s=1}^J\frac{T_s}s -\mu^2\biggr) 
\le \frac{\mu(\bar\mu-\mu)}{n}. 
\end{align}
\end{theorem}
\begin{proof}
See the appendix,
%section~\ref{sec:proofs}, 
where this is proved for a general set of $J$ events, not necessarily
from Gaussian half-spaces.
\end{proof}

The upper bound~\eqref{eq:varub} involves the unknown $\mu$,
so it is not available for planning purpose when we want to select $n$.
The variance and the coefficient of variation,  $\cv(\hat\mu_{\alpha^*})=
\var(\hat\mu_{\alpha^*})^{1/2}/\mu$
can both be bounded in terms of known quantities $\bar\mu$ and $\umu$ 
from~\eqref{eq:mubounds} as follows.
\begin{corollary}\label{cor:cvbounds}
Let $\hat\mu_{\alpha^*}$ be given by \eqref{eq:hatmubal}.
Then
$\var(\hat\mu_{\alpha^*}) \le \bar\mu^2/(4n)$.
If $\umu\ge\bar\mu/2$ then also
$\var(\hat\mu_{\alpha^*}) \le \umu(\bar\mu-\umu)/n$.
Similarly,
\begin{align}\label{eq:cvbounds}
\cv(\hat\mu_{\alpha^*}) \le 
\frac1{\sqrt{n}}\min\Bigl\{
\sqrt{\bar\mu/\umu-1},
\sqrt{J-1}
\,\Bigr\}.
\end{align}
\end{corollary}
\begin{proof}
The claims about $\var(\hat\mu_{\alpha^*})$ follow from
maximizing~\eqref{eq:varub} over $\mu\in[\;\umu,\bar\mu]$.
Next $\cv(\hat\mu_{\alpha^*})^2 = (\bar\mu-\mu)/(n\mu)
=(\bar\mu/\mu-1)/n$. Then~\eqref{eq:cvbounds} follows
because $\mu\ge\umu$ and $\mu\ge\bar\mu/J$.
\end{proof}

%The bounds in 
%do not req
A rare event estimator has bounded relative error 
if $\cv(\hat\mu)$ remains bounded as one takes 
the limit in a sequence of problems \citep[Chapter VI]{asmu:glyn:2007}. 
The sequence is typically one where the event of interest becomes increasingly rare,
for instance as $\mu\to 0$ in the present context.
Corollary~\ref{cor:cvbounds} provides a bounded relative error property 
for ALOE in that limit or indeed  in any sequence of problems where $J/n$ is uniformly bounded.

If the product $H_u(\bsx)H^c_{-u}(\bsx)$ equals one then
it means that $\bsx\in H_j$ if and only if $j\in u$.
We use this to write the union bound in a useful way:
$$
\bar\mu = \sum_{j=1}^J\Pr( H_j(\bsx))
= \sum_{j=1}^J\sum_{u\subseteq1{:}J}
\e\bigl( H_u(\bsx)H^c_{-u}(\bsx)\bigr)1_{j\in u}
=\sum_{s=1}^JsT_s.
$$
That is $\bar\mu = \e( S(\bsx))=\mu\e(S(\bsx)\mid S(\bsx)>0)$ and so we may write~\eqref{eq:varub} as
\begin{align}\label{eq:asprodmoments}
\var(\hat\mu_{\alpha^*}) =
%\frac1n\Bigl( \e(S)\e(S^+)-\mu^2\Bigr),\quad S^+=\begin{cases}1/S, & S>0\\0, &S=0.\end{cases}
\frac{\mu^2}n\Bigl( \e(S\mid S>0)\e(S^{-1}\mid S>0)-1\Bigr).
\end{align}
We will use~\eqref{eq:asprodmoments} in Section~\ref{sec:isamp} to get additional bounds.

\subsection{General Gaussians}\label{sec:generalgaus}

Now suppose that we are given $\bsy\sim\dnorm(\eta,\Sigma)$
and the half-spaces are defined by 
$\gamma_j^\tran \bsy \ge \kappa_j$. 
We assume that $\Sigma$ is nonsingular. If it is not, then 
we can reduce $\bsy$ to a subset of components whose 
variance is nonsingular, and write the other components as linear 
functions of this reduced set. 
We also assume that we can 
afford to take a matrix square root $\Sigma^{1/2}$. 
Now  $\bsx=\Sigma^{-1/2}(\bsy-\eta)\sim\dnorm(0,I)$,
and $\bsy = \eta+\Sigma^{1/2}\bsx$. 
Then the half-spaces are given by 
\begin{align*}
\omega_j^\tran\bsx\ge\tau_j,\quad\text{where}\quad 
\omega_j = 
\frac{\gamma_j^\tran\Sigma^{1/2}}{\sqrt{\gamma_j^\tran\Sigma\gamma_j}},
\quad\text{and}\quad \tau_j 
=\frac{\kappa_j-\gamma_j^\tran\eta}{\sqrt{\gamma_j^\tran\Sigma\gamma_j}},
\end{align*}
for $\bsx\sim\dnorm(0,I)$. 
For rare events, we will have $\kappa_j > \gamma_j^\tran\eta$. 
In some of our motivating contexts one must optimize a cost over $\eta$.
Here we remark that changes to $\eta$ change $\tau_j$ but not $\omega_j$.

\subsection{Sampling algorithms}

We want to sample $\bsx\sim\dnorm(0,I)$ conditionally 
on $\bsx^\tran\omega\ge \tau$ for a unit vector $\omega$
and scalar $\tau$. 
We can use the following steps:
\begin{compactenum}[\quad1)]
\item Sample $\bsz\sim\dnorm(0,I)$.  
\item Sample $u\sim\dustd(0,1)$.  
\item Let $y = \Phi^{-1}(\Phi(\tau)+u(1-\Phi(\tau)))$.  
\item Deliver $\bsx = \omega y + (I-\omega\omega^\tran)\bsz$.  
\end{compactenum} 
These steps can be justified by the analysis in \cite{douc:2010:tr}
who attributes the algorithm to the astrophysics literature.
Step 3 replaces a $\dnorm(0,1)$ distribution for
$y=\bsx^\tran\omega$ by a truncated Gaussian random variable
obtained via inversion.

The algorithm above can be problematic numerically when $\Phi(\tau)$ is 
close to $1$ as it will be for very rare events. 
For instance, in the R language \citep{rlang}, $\Phi(10)$ yields $1$ and 
then $\Phi^{-1}(\Phi(10) + u(1-\Phi(10)))$ yields $\infty$
for any $u$.  Some of our electrical grid examples have $\max_j\tau_j>10^{10}$.
That is, some of the potential failure modes are virtually impossible.

Because $\tau>0$ might be quite large,
we get better numerical stability by sampling $\bsx\sim\dnorm(0,I)$
conditionally on $\bsx^\tran\omega \le -\tau$ and then delivering $-\bsx$. 
The advantage of simulating extreme Gaussians this way was goes back at least to 
\cite{cunn:1969} and may well be older than that.
The steps are as follows:
\begin{compactenum}[\quad1)]
\item Sample $\bsz\sim\dnorm(0,I)$. 
\item Sample $u\sim\dustd(0,1)$. 
\item Let $y = \Phi^{-1}(u\Phi(-\tau))$. 
\item Let $\bsx = \omega y + (I-\omega\omega^\tran)\bsz$. 
\item Deliver $\bsx = -\bsx$. 
\end{compactenum}
Even a very small $u=10^{-12}$
combined with $\tau = 10$ yields 
$$\Phi^{-1}( 10^{-12}\times \Phi(-10)) 
\doteq\Phi^{-1}( 7.62\times 10^{-36}) 
\doteq -12.44 
$$
without any underflow in the R language \citep{rlang}. 
In cases with extremely large $\tau_j$ we will ordinarily get $P_j=0$ and then never
sample conditionally on the corresponding $H_j$.
We compute step 4 via 
$\bsx = \omega y + z - \omega(\omega^\tran\bsz)$ to avoid a potentially 
expensive multiplication $(I-\omega\omega^\tran)\bsz$.

% \subsection{Multivariate $t$}

% The multivariate $t$ distribution is a more heavy-tailed  alternative to the Gaussian.
% A multivariate $t$ random variable $\bsy$ with center $\eta\in\real^d$, positive definite scale matrix
% $\Sigma\in\real^{d\times d}$ and degrees of freedom $\nu>0$ can be represented via
% $$
% \bsy = \eta + \frac{\Sigma^{1/2}\bsx}{\sqrt{W/\nu}}
% $$
% where $W\sim\chi^2_{(\nu)}$ independently of $\bsx\sim\dnorm(0,I)$.
% The half-space $\gamma_j^\tran\bsy\ge\kappa_j$ correspond to 
% $\omega_j^\tran\bsx\ge\tau_j$ where
% $$
% \omega_j = \frac{\gamma^\tran_j\Sigma^{1/2}}{\sqrt{\gamma_j^\tran\Sigma\gamma_j}},
% \quad\text{and},\quad
% \tau_j = \frac{\kappa_j-\gamma_j^\tran\eta}{\sqrt{\gamma_j^\tran\Sigma\gamma_j}}\sqrt{W/\nu}.
% $$

\section{Importance sampling properties}\label{sec:isamp}

As shown in the Appendix, Theorem~\ref{thm:muvar}
holds more generally than the Gaussian case.
In this more general setting,
we have $J$ events, $H_j$, on a common sample space $\cx$
where $\bsx\in\cx$ has probability density $p$.
Event $H_j$ has probability $P_j$. 
As before, we want $\mu = \Pr(H)$ where $H=\cup_jH_j$
and the union bound  is $\mu\le\bar\mu = \sum_jP_j$.  
We assume that $0<\bar\mu<\infty$.
The upper bound only has to be checked if $J=\infty$.
If $\bar\mu=0$, then we know $\mu=0$ without any sampling.

When we sample, we  ensure that
at least one rare event takes place every time, by first picking
an event $H_j$ with probability proportional to $P_j$.
Then we sample $\bsx\in\cx$ conditionally on $H_j$
and find $S(\bsx) = \sum_{\ell=1}^J H_\ell(\bsx)$, the total number of events
that occur.  This includes $H_j$ and so our sample values
always have $S(\bsx_i)\ge1$.
The importance sampling estimate $\hat\mu_{\alpha^*}$
averages $\bar\mu/S(\bsx_i)$ over $n$ independent replicates.
As in the prior section, we use 
$$T_s =\Pr(S(\bsx)=s) = \int_{\real^d}1\{S(\bsx)=s\}p(\bsx)\rd\bsx,$$
for the probability of exactly $s$ events happening.
Then the variance of $\hat\mu$ is given by~\eqref{eq:varub}.

The optimal importance sampling distribution for estimating
$\mu$ is uniform on $H=\{\bsx\mid H(\bsx)=1\}$.
Sampling from this distribution would yield an
estimate with variance zero.  Not surprisingly, we are seldom able
to do that in applications.  The ALOE sampler takes $\bsx\in\cx$
with probability proportional to $S(\bsx)$, so it has support set $H$.

We think that many applications will have events $H_j$ that rarely
co-occur. In that case $S(\bsx)$ is nearly constant at $1$ 
for $\bsx\in H$, and the ALOE sampler is close to the optimal importance sampler.
Other applications may have a few near duplicated events that co-occur
often.  One extreme setting has a common cause that triggers
all $J$ events at once and those events almost never arise
outside of that common situation.   
In that case $S(\bsx)$ is again nearly constant on $H$,
this time usually equal to $J$, and ALOE is again nearly optimal. 

The variance bound $\mu(\bar\mu-\mu)/n$ from~\eqref{eq:varub} can be conservative.  
It stems from $T_s/s\le T_s$, when $s\ge1$.  
If $\Pr_{\alpha^*}( S>1)$ is appreciably large  
then the variance can be meaningfully less than that bound.  
%For rare events we might expect that  
%$\Pr(S=1\mid S\ge1)$ is close to one under the nominal distribution.  
% Because $\bar\mu \le J\mu$ we also have  
% \begin{align}\label{eq:jbound}
% \var(\hat\mu_{\alpha^*})\le (J-1)\frac{\mu^2}n.  
% \end{align} 
% Equality in~\eqref{eq:jbound} only happens if $\bar\mu= J\mu$.  
% Then the $J$ 
% events are almost surely the same and $\var(\hat\mu_{\alpha^*})=0$.  
% Therefore strict inequality always holds in equation~\eqref{eq:jbound}.  
% The trivial exception from $\mu=0$ is not possible because we have
% $\bar\mu>0$.
We can improve the variance bound %sharpen the bound~\eqref{eq:jbound}
by using the following lemma.
\begin{lemma}\label{lem:prodmomentbound}
Let $S$ be a random variable supported on $\{1,2,\dots,J\}$
%for $1\le J<\infty$.  
for $J\in\natu$.
Then
\begin{align}\label{eq:prodmomentbound}
\e(S)\e(S^{-1})
 \le \frac{J+J^{-1}+2}4
\end{align}
with equality if and only if $S\sim\dustd\{1,J\}$.
\end{lemma}
\begin{proof}
See the appendix. %section~\ref{sec:proofs}.
\end{proof}

Lemma~\ref{lem:prodmomentbound} tells us that for $J\ge2$, our worst case setting
is one where half of the time that one or more events happen, exactly
one happens and half of the time, all $J$ of them happen.
While that is not plausible for Gaussian $\bsx$ and large $J$
it can indeed happen for
combinatorial enumeration problems like those of \cite{karp:luby:1983}.
% We might have guessed this, because with $\bar\mu=\e(S)$ fixed
% at some level, the variance of $S^{-1}$ is maximized when $S^{-1}$
% takes only the two extreme values $1$ and $J^{-1}$ with equal probability.
% The situation is a bit more complicated because $\e(S)$ and $\e(S^{-1})$
% both depend on the distribution of $S$ and we need to bound the
% product $\e(S)\e(S^{-1})$, where both expectations are under
% $q_{\alpha^*}$.
From Theorem~\ref{thm:varub} and Lemma~\ref{lem:prodmomentbound}, we get
\begin{align}\label{eq:lemupper}
\var(\hat\mu_{\alpha^*}) =
\frac{\mu^2}n\biggl(
\biggl(\,\sum_{s=1}^Js\frac{T_s}\mu\biggr)
\biggl(\,\sum_{s=1}^Js^{-1}
\frac{T_s}\mu\biggr)
-1\biggr)
 \le 
\frac{\mu^2}n
\frac{J+J^{-1}-2}4,
\end{align}
because $T_s/\mu$ is a probability distribution on $\{1,2,\dots,J\}$. 
%This bound describes relative error without reference to
%the union bound $\bar\mu$.

% When we want an upper bound in terms of $\bar\mu$, then
% the right hand side of~\eqref{eq:varub} is maximized over $\mu$
% at $\mu =\bar\mu/2$, yielding
% \begin{align}\label{eq:crudeupper}
% \var( \hat\mu_{\alpha^*})\le \frac{\bar\mu^2}{4n}.
% \end{align}
% Using $\mu\le\bar\mu$ in~\eqref{eq:lemupper} yields
% $$
% \var( \hat\mu_{\alpha^*})\le \frac{\bar\mu^2}{4n}(J+J^{-1}-2).
% $$
% This is only an improvement on~\eqref{eq:crudeupper} when $J=2$.
% The inequality $\mu=\bar\mu$ cannot actually hold
% under the conditions where equality holds in~\eqref{eq:lemupper}.

Sometimes we are interested in the probability of sub-events of $H$.
Let $f(\bsx)$ be supported on $H$
and define
$\nu(f) = \nu=\int f(\bsx)p(\bsx)\rd\bsx = 
\int_H f(\bsx)p(\bsx)\rd\bsx$.
We may use ALOE, via
$$
\hat\nu=
\frac{\bar\mu}n\sum_{i=1}^n\frac{f(\bsx_i)}{S(\bsx_i)},
\quad\bsx_i\simiid q_{\alpha^*}.
$$
Then by the same arguments used in the Appendix,
$$
\e(\hat\nu) =\nu
\quad\text{and}\quad
\var(\hat\nu)=
\frac1n\biggl(
\bar\mu
\int_H \frac{f(\bsx)^2p(\bsx)}{S(\bsx)}\rd\bsx
-\nu^2\biggr).
$$
If $f(\bsx)\in\{0,1\}$, then
$\var(\hat\nu) \le \nu(\bar\mu-\nu)/n$.
That is, when $f$ describes a rare event that can only
occur if one or more of the $H_j$ also occur, we can reduce
its Monte Carlo variance from $\nu(1-\nu)/n$
to at most $\nu(\bar\mu-\nu)/n$, in cases where $\bar\mu<1$.

\section{Comparisons}\label{sec:comparison}

Here we consider some numerical examples comparing 
ALOE to {\tt pmvnorm}
from the R package {\tt mvtnorm} \citep{mvtnorm}. 
This package can make use of special properties of the Gaussian 
distribution, and it works in high dimensions. 

We begin by describing {\tt mvtnorm} based on 
\cite{genz:bret:2009} and a personal communication 
from Alan Genz. The program computes 
$$\Pr( \bsa\le\bsy\le \bsb)\equiv 
\Pr( a_j \le y_j \le b_j ,\, j=1,\dots,d)$$
for $\bsy\sim\dnorm(\eta,\Sigma)$,
where $-\infty \le a_j \le b_j \le\infty$ for $j=1,\dots,d$,
and $\Sigma$ can be rank deficient. 
We can use it to compute $\mu = \Pr( \sum_{j=1}^J\one\{\omega_j^\tran\bsx\ge \tau_j\}>0)$
for $\bsx\sim\dnorm(0,I)$ via 
$$
1-\mu 
= \Pr( \Omega^\tran\bsx\le\ct) 
= \Pr( \bsy\le\ct),\quad\bsy\sim\dnorm(0,\Omega^\tran\Omega). 
$$

The code can handle dimensions up to $1000$. In our context, that 
means at most $J=1000$ half-spaces. The dimension $d$ can be higher. 
The related {\tt pmvt} function handles multivariate $t$ random variables. 
The code has three different algorithms in it. 
One from \cite{genz:2004} handles 
 two and three dimensional semi-infinite regions, one from \cite{miwa:etal:2003}
is for dimensions up to $20$ and the rest are handled 
by an algorithm from~\cite{genz:bret:2009}. This latter algorithm uses 
a number of methods. It uses randomized Korobov lattice rules 
as described by \cite{cran:patt:1976} for the first $100$
dimensions, in conjunction with antithetic sampling. There 
are usually $8$ randomizations. For more than $100$ dimensions 
it applies a method from \cite{nied:1972}. 
There are a series of increasing sample sizes in use, and the method 
provides an estimated error ($3.5$ standard errors) based on the randomization. 
The approach is via sequential conditional sampling, after strategically 
ordering the variables (e.g., putting unconstrained ones first). 
The R package calls a FORTRAN program for the computation, so it is very fast. 
We use the default implementation which uses up to $25{,}000$ quadrature points. 

The main finding in comparison to \cite{genz:bret:2009} is that importance sampling is more effective 
when the polytope of interest is the complement of a rare event. 
This is not meant to be a criticism of {\tt pmvnorm}. That code was not 
specifically designed to compute the complement of a rare event. 
The comparison is relevant because we are not aware of alternative code tuned for the 
high dimensional rare event cases that we need, 
and {\tt pmvnorm} is a well regarded and widely available general solution, 
that seemed to us like the best off-the-shelf tool. 

\cite{botev2015efficient} provide a competing method to \cite{genz:bret:2009} 
for  estimating polytope probabilities 
with Gaussian and $t$-distributed data. Like \cite{genz:bret:2009} 
they address problems where   $\bsx\in\cp$ is the rare event,
not $\bsx\not\in \cp$.  We have not compared their method numerically
to ALOE but we expect that like \cite{genz:bret:2009}, it will dominate
ALOE when the event is not rare but not when the event is rare.

A reviewer asked us to compare our method to the
recent work in \cite{ahn:kim:2018} on computing 
expectations over a union of half-spaces.
Their approach is to pick a unit vector $\delta$ uniformly at
random in $\sphd=\{\bsx\in\real^d\mid \bsx^\tran\bsx=1\}$
and average $\Pr(\bsx\in H)$ over the
line $\{y\delta\mid y\in\real\}$.  This line extends in the positive $y$
reaching the set $H$ at distance 
$\min_j\tau_j/\omega_j^\tran\delta_i$ taking the minimum over $j$
with $\omega^\tran\delta>0$.  Should that set of $j$ be empty, the
line never reaches $H$ in the positive $y$ direction.
It has a similarly defined extent in the negative $y$ direction.
The scale $y$ is a symmetric random variable with $y^2\sim\chi^2_{(d)}$
because $\Vert\bsx\Vert^2\sim\chi^2_{(d)}$.
Putting these together we find that their 
directional simulation estimator is
\begin{align}\label{eq:muak}
\hat\mu_{\ak} = \frac1{2n}\sum_{i=1}^n 
\max_{1\le j\le J} \bar G_d\Bigl( \frac{\tau_j^2}{(\omega^\tran\delta_i)^2}\Bigr)\one\{\omega^\tran\delta_i>0\}
+\max_{1\le j\le J} \bar G_d\Bigl( \frac{\tau_j^2}{(\omega^\tran\delta_i)^2}\Bigr)\one\{\omega^\tran\delta_i<0\}
\end{align}
where $\bar G_d$ is one minus the cumulative distribution function of $\chi^2_{(d)}$.
They use their estimator on some ellipsoidally symmetric distributions
generalizing the Gaussian.
We ran directional sampling on the two examples
described next, and it was extremely good on
one of them and extremely inaccurate on the other.

\subsection{Circumscribed polygon}

Let $\cp(J,\tau)\subset\real^2$  be the regular polygon 
of $J\ge3$ sides circumscribed around the circle of radius $\tau>0$. 
This polygon is the intersection of $H_j^c$ where
$H_j=\{\bsx\in\real^2\mid \omega_j^\tran\bsx\ge\tau\}$
where $\omega_j^\tran =(\sin(2\pi j/J), \cos(2\pi j/J))$, for $j=1,\dots,J$.
We want $\mu = \Pr( \bsx\in \cp^c)$ for $\bsx\sim\dnorm(0,I)$. 
Here we know that $\mu\le\Pr( \chi^2_{(2)}\ge \tau^2) = \exp(-\tau^2/2)$. 
Also, the gap between the circle of radius $\tau$ and the circumscribed 
polygon has area $G(J,\tau) = (J\tan(\pi/J)-\pi)\tau^2$. 
The bivariate Gaussian density in this gap is at most 
$\exp(-\tau^2/2)/(2\pi)$. 
Therefore 
$$
\Pr(\bsx\in \cp^c) \ge 
\exp(-\tau^2/2)-G(J,\tau) \exp(-\tau^2/2)/(2\pi) 
$$
that is 
$$
%1-\frac{G(J,1)\tau^2}{2\pi}
%\le\frac{\Pr(\bsx\in \cp^c)}{\exp(-\tau^2/2)} \le 1 
1 \ge\frac{\Pr(\bsx\in \cp^c)}{\exp(-\tau^2/2)} 
\ge 1-\frac{G(J,1)\tau^2}{2\pi}
\doteq  
%1-\frac{\pi^3\tau^2/(3J^2)}{2\pi}
 1-\frac{\pi^2\tau^2}{6J^2},
$$
for large $J$. 

For $J=360$ and $\tau=6$, we have $\mu \le \exp(-18)\doteq 1.52\times10^{-8}$. 
The lower bound is about $0.9995$ times the upper bound, so we treat the 
upper bound as exact. 
Figure~\ref{fig:sect6} shows histograms of $100$ simulations of $\hat\mu/\mu$
using ALOE and using pvnorm. In this case ALOE is much more accurate. 
The mean square relative error $\e( (\hat\mu/\mu-1)^2)$ is about $800$-fold 
smaller for ALOE than pvnorm.  We also see that pvnorm has high positive
skewness and the histogram of estimates has most of its mass well below the mean.

Table~\ref{tab:othertau} shows summary results for this problem with different 
values of $\tau$.  We see that pvnorm is superior when the event is not 
rare but ALOE is superior for rare events. 
The large error for pvnorm with $\tau =5$ stemmed from a small number 
of outliers among the $100$ trials. 

\begin{figure}[t]
\includegraphics[width=.9\hsize]{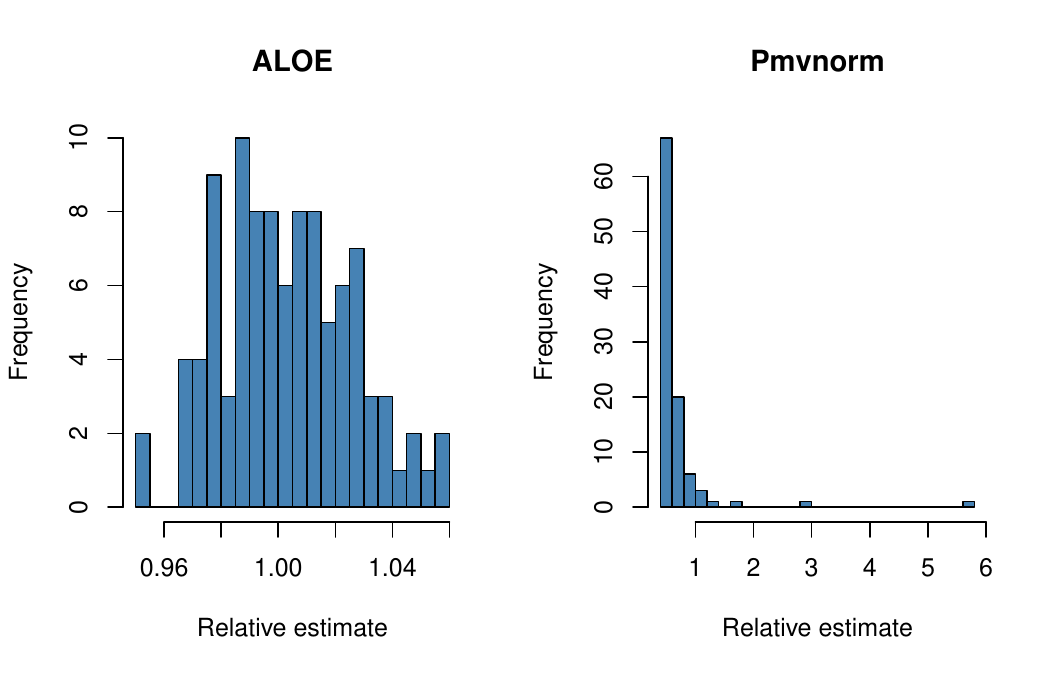}
\caption{\label{fig:sect6}
Results of $100$ estimates of the $\Pr(\bsx\not\in\cp(360,6))$,
divided by $\exp(-6^2/2)$. Left panel: ALOE.
Right panel: pmvnorm. 
}
\end{figure}

\begin{table}
\centering 
% > source("figsect.R");sectnumbers()[,c(1,2,5,6)]
%      tau      pup     phat      pmv 
% [1,]   2 1.35e-01 0.000399 9.42e-08 
% [2,]   3 1.11e-02 0.000451 9.24e-07 
% [3,]   4 3.35e-04 0.000549 2.37e-02 
% [4,]   5 3.73e-06 0.000600 1.81e+00 
% [5,]   6 1.52e-08 0.000543 4.39e-01 
% [6,]   7 2.29e-11 0.000559 3.62e-01 
% [7,]   8 1.27e-14 0.000540 1.34e-01 
\begin{tabular}{ccccc}
\toprule 
    $\tau$ & $\mu$ &$\e((\hat\mu_{\mathrm{ALOE}}/\mu-1)^2)$
&$\e((\hat\mu_{\mathrm{MVN}}/\mu-1)^2)$  \\
\midrule 
2 & 1.35$\times10^{-01}$ & 0.000399 & 9.42$\times10^{-08}$  \\
3 & 1.11$\times10^{-02}$ & 0.000451 & 9.24$\times10^{-07}$  \\
4 & 3.35$\times10^{-04}$ & 0.000549 & 2.37$\times10^{-02}$  \\
5 & 3.73$\times10^{-06}$ & 0.000600 & 1.81$\times10^{+00}$  \\
6 & 1.52$\times10^{-08}$ & 0.000543 & 4.39$\times10^{-01}$  \\
7 & 2.29$\times10^{-11}$ & 0.000559 & 3.62$\times10^{-01}$  \\
8 & 1.27$\times10^{-14}$ & 0.000540 & 1.34$\times10^{-01}$  \\
\bottomrule 
\end{tabular}
\caption{\label{tab:othertau}
Results from $100$ computations of $\Pr( \bsx\not\in \cp(360,\tau))$ for 
various $\tau$. The true mean $\mu$ is very nearly $\exp(-\tau^2/2)$. 
Importance sampling is more accurate for large $\tau$ (rare events),
while {\tt pmvnorm} is more accurate for small $\tau$. 
}
\end{table}

The upper bound in equation~\eqref{eq:varub} is 
 $\var(\hat\mu) \le \mu(\bar\mu-\mu)/n$,
from which $\e( (\hat\mu/\mu-1)^2) \le (\bar\mu/\mu-1)/n$. 
For $\tau=6$ this yields about $0.022$, which is over $20$
times the actual mean squared relative error from Table~\ref{tab:othertau}. 

It is possible that this example is artificially easy for importance sampling,
due to the symmetry.  Whichever half-space $H_j$ we sample, the distribution of 
overlapping half-spaces $H_k$ for $k\ne j$ is the same. Two half-spaces differ
from $H_j$ by a one degree angle, two differ by a two degree angle and so on. 
To get a more varied range of overlap patterns, we replaced angles 
$2\pi j/360$ by angles $2\pi\times p(j)/360$ where $p(j)$ is the $j$'th 
prime among integers up to $360$.  There are $72$ of them, of which the 
largest is $359$. With $\tau=6$ and $100$ replications using $n=1000$
points in importance sampling, we have variance of $\hat p/\exp(-18)$
equal to $0.00077$. The comparable figure for {\tt mvtnorm} is $8.5$. 
There were a few outliers there including one that was more than $6$ times 
the union bound.  The gap between the prime angle polygon and the 
inscribed circle is larger than the one formed by the full polygon. 
Pooling all the importance sampling runs leaves an 
estimate of about $0.94\times \exp(-18)$ for $\mu$.  In this example, we 
see importance sampling working quite well without symmetry. 

The estimator $\hat\mu_{\ak}$ is much better than both ALOE and {\tt pmvnorm} for this problem.
With only a sample of $n=100$ it reached essentially double precision
accuracy, with a standard error of about $5\times 10^{-16}$ on the
symmetric polygon.  For the problem using prime number angles
the standard error was about $3\times 10^{-12}$.
This problem is even more artificially easy for directional simulation method
than the polygon is for ALOE.  The distance from the origin to $H$ is
nearly constant over all angles.

\subsection{High dimensional half-spaces}

The previous example was low dimensional and each of the 
half-spaces sampled had numerous similar ones, differing in angle by a small number 
of degrees. Thus $\mu$ was quite a bit smaller than $\bar \mu$. 
Here we consider a high dimensional setting where the half-spaces 
have less overlap. 

Two uniform random unit vectors $\omega_1$ and $\omega_2$ 
in $\real^d$ are very likely 
to be nearly orthogonal for large $d$. Then 
$\bsx^\tran\omega_j>\tau_j$  are nearly independent events. 
For independent events, we would have 
$$
\Pr(\bsx\not\in\cp) = 1-\prod_{j=1}^J(1-P_j). 
$$
To make $\bsx\not\in\cp$ a rare event, the $P_j$ must be small and 
then the probability above will be close to the union bound. 
Theorem~\ref{thm:varub} predicts good performance 
for importance sampling here. 

For this test $200$ sample problems were constructed. 
The dimensions were chosen by $d\sim\dustd\{20,50,100,200,500\}$. 
Then there were $J\sim\dustd\{d/2,d,2d\}$ constraints chosen 
with uniform random unit vectors $\omega_j\in\real^d$. 
The threshold $\tau$ was chosen so that $\log_{10}$ of 
the union bound was $\dustd[4,8]$, followed by rounding to two significant 
figures. Then $\hat\mu$ was computed by importance sampling 
with $n=1000$ samples, and by {\tt pmvnorm}. 
Figure~\ref{fig:topes} shows the results. The ALOE sampling value 
was always very close to the union bound which in turn is essentially 
equal to what one would see for independent events. The values from 
{\tt pmvnorm} were usually too small but sometimes far too large,
orders of magnitude larger than the union bound.  By construction the intersection 
probabilities are quite rare.  In importance sampling, $77.5$\% of the 
simulations had no intersections among $1000$ trials and the 
others had only a few intersections. Therefore it is clear that 
the probabilities should be close to the union bounds here. 
%We also see evidence of strong positive skewness for {\tt pmvnorm}.

\begin{figure}[t]
\includegraphics[width=.9\hsize]{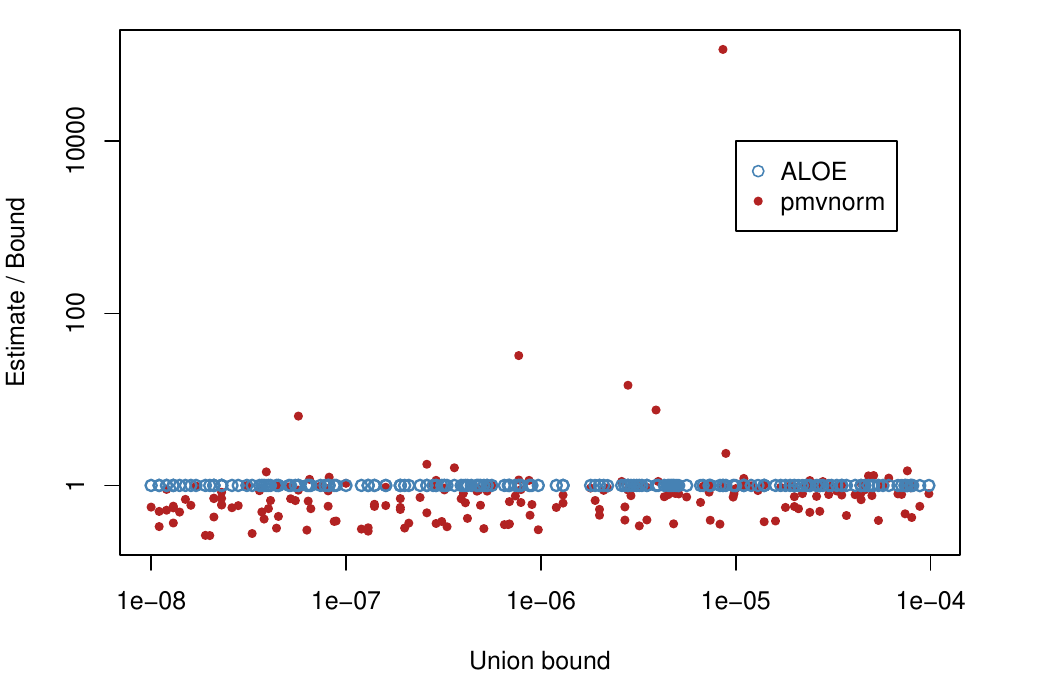}
\caption{\label{fig:topes}
Results of $200$ estimates of the $\mu$
for varying high dimensional problems with nearly independent events. 
}
\end{figure}

We also ran the directional method on these high-dimensional half space
problems.  It did very poorly because the set $H$ only comes close to
the origin at a tiny proportion of the unit vectors $\delta\in\sphd$.
The estimate $\hat\mu_{\ak}$ was usually smaller
than the known lower bound $\umu = \max_{1\le j\le J}\Phi(-\tau_j)$,
sometimes much smaller, in one instance below $10^{-30}\umu$.
Two hundred results are presented in Figure~\ref{fig:direction}.
That estimator was also larger than the known upper bound $\bar\mu$
by as much as $40$-fold in some simulation.  It was less severe
in that regard than {\tt pmvnorm}.
Because it performed so poorly we did
not implement it on the power systems problem in the next section.

\begin{figure}[t]
\includegraphics[width=.9\hsize]{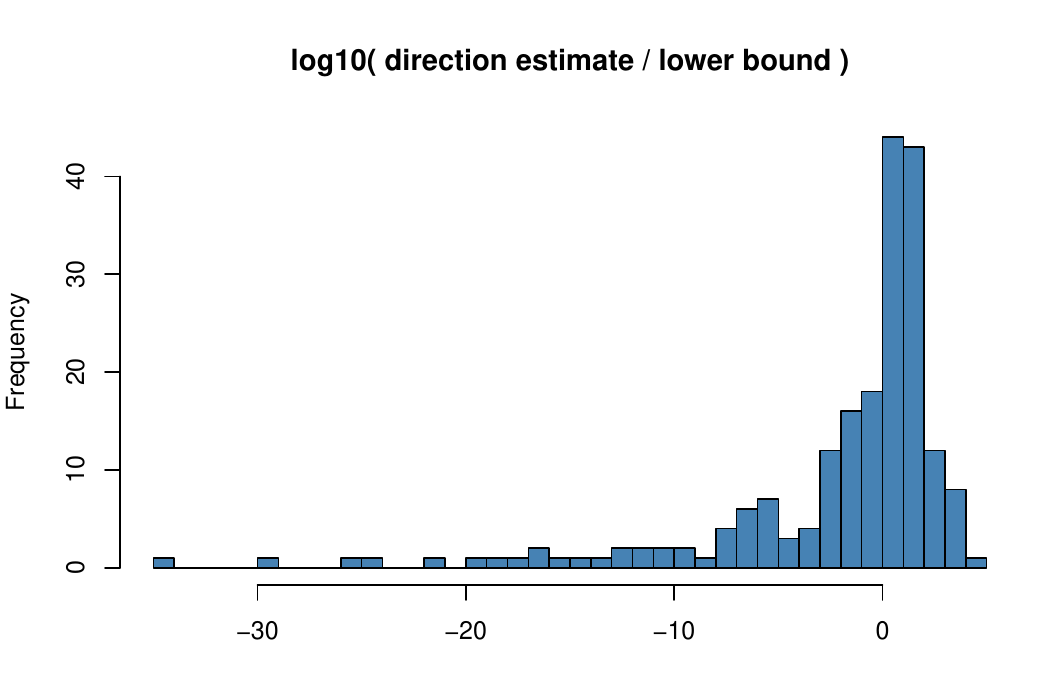}
\caption{\label{fig:direction}
Results of $200$ estimates of $\log_{10}(\mu_{\ak}/\umu)$
for varying high dimensional problems with nearly independent events. 
}
\end{figure}

\section{Power system infeasibility}\label{sec:power}

\subsection{Model}
Our power system models are based on a network of $N$ nodes, called busses.
Some busses put power into the network and others consume power. 
The $M$ edges in the network correspond to power lines between busses.
The network is ordinarily sparse, with $M$ a small multiple of $N$.

The power production at bus $i$ is $p_i$, with negative values
indicating consumption.
For some busses, $p_i$ is tightly controlled and deterministic
in the relevant time horizon. Other busses  have random
$p_i$ corresponding, for example, to variable consumption
levels, that we treat as independent. 
Busses corresponding to wind farms have random power production
levels with meaningfully large correlations.
Our models contain one special bus $S$, called the slack bus, at which the power
is $p_S =-\sum_{i\ne S}p_i$. The total power in the system is
zero because transmission power losses are ignored in the DC 
approximation that  we use.

The power at all busses can be represented by the vector
$p = (p_F^\tran,p_R^\tran,p_S)^\tran$ corresponding to fixed busses,
ordinary random busses (including any correlated ones) and the slack bus.
There are $N_F$ fixed busses, $N_S=1$ slack bus and $N_R=N-N_F-N_S$
random busses apart from the slack bus.  We will use
$1_R$ to denote a column vector of $N_R$ ones, and $I_R$ to
denote the identity matrix of size $N_R$ and similarly for $1_F$ and $I_F$.

The power $p_i$ at bus $i$ must satisfy the constraints
\begin{align}\label{eq:powercon}
 \underline p_i \le p_i \le \overline p_i.
\end{align}
The vector $p$ has a Gaussian distribution, determined entirely by the random
components $p_R\sim\dnorm(\eta_R,\Sigma_{RR})$. 
Therefore in the present context, $p_R$ is
the Gaussian random variable $\bsx$ from Section~\ref{sec:gaussian}.
The fixed components
satisfy $p_F=\eta_F$ and then the slack bus satisfies
$p_S\sim\dnorm(\eta_S,\Sigma_{SS})$ where
$\eta_S = -1_R^\tran \eta_R-1_F^\tran\eta_F$ and
$\Sigma_{SS} = 1_R^\tran\Sigma_{RR}1_R$.
Because all of the randomness comes from $p_R$, we will abbreviate
$\Sigma_{RR}$ to $\Sigma$.

The node to node inductances in the network form
a Laplacian matrix $B$ where $B_{ij}\ne0$ if busses $i$
and $j$ are connected with $B_{ii} = -\sum_{j\ne i}B_{ij}$
(up to rounding).
The Laplacian is symmetric and has one eigenvalue of zero for
a connected network.
It has a pseudo-inverse $B^+$. We partition $B$ and $B^+$ as follows
$$
B = 
\begin{pmatrix}
B_{RR} & B_{RF} & B_{RS}\\
B_{FR} & B_{FF} & B_{FS}\\
B_{SR} & B_{SF} & B_{SS}\\
\end{pmatrix}, \quad\text{and}\quad
B^+ = 
\begin{pmatrix}
B^{RR} & B^{RF} & B^{RS}\\
B^{FR} & B^{FF} & B^{FS}\\
B^{SR} & B^{SF} & B^{SS}\\
\end{pmatrix}.
$$
We also group $B^+$ into three sets of columns via
$B^+ =\begin{pmatrix} B^{\sumdot R} & B^{\sumdot F} & B^{\sumdot S}
\end{pmatrix}$.

The phase at bus $i$ is denoted $\theta_i$.
In our DC approximation of AC power flow,
the phases approximately
satisfy $B\theta =p$.  Given the power vector $p$, we take
$$
\theta = B^+p
=\begin{pmatrix}
B^{RR} & B^{RF} & B^{RS}\\
B^{FR} & B^{FF} & B^{FS}\\
B^{SR} & B^{SF} & B^{SS}\\
\end{pmatrix}
\begin{pmatrix}
p_R\\p_F\\p_S
\end{pmatrix}.
$$
The phase constraints on the network are
\begin{align}\label{eq:phasecon}
|\theta_i-\theta_j| \le \bar\theta_{ij},\quad\text{for $i\ne j$ and  $B_{ij}\ne0$}.
\end{align}
In our examples, all $\bar\theta_{ij} = \bar\theta$ for a single value
$\bar\theta$ such as $\pi/6$ or $\pi/4$.

Let $D\in\{-1,1\}^{M\times N}$ be the incidence matrix.  Each edge in
the network is represented by one row of $D$ with an entry of $+1$ for
one of the busses on that edge and $-1$ for the other. The phase constraints
are $|D\theta|\le \bar\theta$ componentwise.
Now
$$
D\theta = 
D
\begin{pmatrix}
B^{\sumdot R} & B^{\sumdot F} & B^{\sumdot S}
\end{pmatrix}
\begin{pmatrix}
p_R\\p_F\\p_S 
\end{pmatrix}=
D\bigl(
   B^{\sumdot R} p_R 
+ B^{\sumdot F} p_F 
+ B^{\sumdot S} p_S 
\bigr). 
$$

The constraint that $D\theta\le \bar\theta$ for every edge $ij$  can be written
\begin{align*}
DB^{\sumdot R}p_R \le \bar\theta 
-DB^{\sumdot F}p_F
-DB^{\sumdot S}p_S.
\end{align*}
Now $p_S = -1_R^\tran p_R-1_F^\tran p_F$
and $p_F=\eta_F$, so the constraint on $p_R$ is
\begin{align}\label{eq:phasecon1} 
D\bigl( B^{\sumdot R}-B^{\sumdot S}1_R^\tran\bigr) p_R  
\le \bar\theta -  
D\bigl( B^{\sumdot F}-B^{\sumdot S}1_F^\tran\bigr)\eta_F.  
\end{align} 
We also have constraints $D\theta\ge-\bar\theta$ 
which can be written
\begin{align}\label{eq:phasecon2}
D\bigl(B^{\sumdot S}1_R^\tran- B^{\sumdot R}) p_R \bigr)
\le \bar\theta + 
D\bigl( B^{\sumdot F}-B^{\sumdot S}1_F^\tran\bigr)\eta_F. 
\end{align}

Equations~\eqref{eq:phasecon1}  and~\eqref{eq:phasecon2}  supply $2M$
constraints on the random vector $p_R$.
We have also the two slack bus constraints
$p_S\le \overline p_S$ and $-p_S \le -\underline p_S$, that is
\begin{align}\label{eq:slackcon}
-1_R^\tran p_R &\le \overline p_S +1_F^\tran \eta_F\quad\text{and}\quad
1_R^\tran p_R \le -\underline p_S +1_F^\tran \eta_F.
\end{align}
Finally, there are individual constraints on the random busses
$$
p_R \le \overline p_R,\quad\text{and}\quad -p_R \le -\underline p_R,
$$
componentwise.

When we combine all of the constraints on $p_R$, we
get a matrix $\Gamma$ with rows $\gamma_j$
and a vector of upper bounds $\ck$ with entries $\kappa_j$
for which the constraints are 
$\Gamma\times p_R\le \ck$ componentwise. Here those matrices are
$$
\Gamma =
\begin{pmatrix}
\phm I_R\\[.7ex]
-I_R\\[.7ex]
\phm1_R^\tran\\[.7ex]

-1_R^\tran\\[.7ex]
\, \phm D(B^{\sumdot R}-B^{\sumdot S}1_R^\tran)\,\\[.7ex]
\,  -      D(B^{\sumdot R} -B^{\sumdot S}1_R^\tran)\,
\end{pmatrix},\quad\text{and}\quad
\ck = 
\begin{pmatrix}
\phm\overline p_R\\[.7ex]
      -\underline p_R\\[.7ex]
-\underline p_S+1_F^\tran\eta_F\\[.7ex]
\phm\overline p_S+1_F^\tran\eta_F\\[.7ex]
\bar\theta -D\bigl( B^{\sumdot F}-B^{\sumdot S}1_F^\tran\bigr)\eta_F\,\\[.7ex]
\bar\theta +D\bigl( B^{\sumdot F}-B^{\sumdot S}1_F^\tran\bigr)\eta_F\,
\end{pmatrix}.
$$

These are the linear constraints on $p_R$. 
There are $2M$ phase constraints and there are two constraints
for all of the non-fixed busses, including the slack bus.
These constraints can be turned into constraints
$\Omega$ and $\ct$ on a $\dnorm(0,I_R)$
vector as described in Section~\ref{sec:generalgaus}.

\subsection{Examples}

We considered several model electrical grids  included in the MATPOWER
distribution \citep{zimm:muri:thom:2011}.
In each case we modeled violations of the phase constraints,
and used $n=10{,}000$ samples.
For some cases we found that, under our model, phase constraint
violations were not rare events.
In some other cases, the rare event probability was dominated
by one single phase condition:   $\umu=\max_{j=1}^JP_j\approx\sum_{j=1}^JP_j=\bar\mu$.  
For cases like this there
is no need for elaborate computation because we 
know $\mu$ is within a narrow interval
$[\umu,\bar\mu]$.
The interesting cases were of rare events not dominated by a single failure mode.
We investigate two of them.

The first is the Polish winter peak grid of 2383 busses.
There were $d=326$ random (uncontrolled) busses and $J=5772$ phase constraints.
We varied $\bar\omega$ as shown in Table~\ref{tab:winterpeak}.
For $\bar\omega=\pi/7$ constraint violations are not very rare.  At $\bar\omega =\pi/4$
they are quite rare. The estimated coefficient of variation is nearly constant over this range.

\begin{table}\centering 
\begin{tabular}{ccccc}
\toprule 
$\bar\omega$ & $\hat \mu$ & se$/\hat\mu$ & $\umu$ & $\bar\mu$\\
\midrule 
$\pi/4$ & $3.7\times 10^{-23}$ & $0.0024$ & $3.6\times10^{-23}$ & $4.2\times 10^{-23}$ 
\\
$\pi/5$ & $2.6\times 10^{-12}$ & $0.0022$ & $2.6\times10^{-12}$ & $2.9\times 10^{-12}$
\\
$\pi/6$ & $3.9\times 10^{-07}$ & $0.0024$ & $3.9\times10^{-07}$ & $4.4\times 10^{-07}$
\\
$\pi/7$ & $2.0\times 10^{-03}$ & $0.0027$ & $2.0\times 10^{-03}$ & $2.4\times10^{-03}$
\\
\bottomrule 
\end{tabular}
\caption{\label{tab:winterpeak}
Rare event estimates for the winter peak grid. 
$\bar\omega$ is the phase constraint, $\hat\mu$ is the ALOE estimate,
se is the estimated standard error, $\umu$ is the largest single
event probability and $\bar\mu$ is the union bound.
}
\end{table}

The second interesting case is the Pegase 2869 model of \cite{flis:panc:capi:wehe:2013}.  
This has $d=509$ uncontrolled
busses and $J=7936$ phase constraints. It is described as ``power flow
for a large part of the European system''.
The results are shown in Table~\ref{tab:pegase2869}.
We include an unrealistically large bound $\bar\omega=\pi/2$
in that table, to test the limits of our approach.
For $\bar\omega=\pi/2$, the standard error given is zero. One half-space
was sampled $9408$ times, another was sampled $592$ times but in
no instance were there two or more phase violations. The estimate reverts
to the union bound. Getting $0$ doubletons ($S=2)$ among $n=10{,}000$ tries
is compatible with the true probability of a doubleton being as high as $3/n$.
Even if $T_1=.9997$ and $T_2=.0003$ then we would have $\mu = (1-(3/2)\times 10^{-4})\bar\mu$
instead of $\bar\mu$.  We return to this issue in the discussion.

\begin{table}\centering 
\begin{tabular}{ccccc}
\toprule 
$\bar\omega$ & $\hat \mu$ & se$/\hat\mu$ & $\umu$ & $\bar\mu$\\
\midrule 
$\pi/2$ & $3.5\times 10^{-20}$ & $0^*$ & $3.3\times 10^{-20}$ & $3.5\times 10^{-20}$
 \\
$\pi/3$ & $8.9\times 10^{-10}$ & $5.0\times10^{-5}$ & $7.7\times 10^{-10}$ & $8.9\times 10^{-10}$
 \\
$\pi/4$ & $4.3\times 10^{-06}$ & $1.8\times10^{-3}$ & $3.5\times 10^{-06}$ & $4.6\times 10^{-06}$
 \\
$\pi/5$ & $2.9\times 10^{-03}$ & $3.5\times10^{-3}$ & $1.8\times 10^{-03}$ & $4.1\times 10^{-03}$
 \\
\bottomrule 
\end{tabular}
\caption{\label{tab:pegase2869}
Rare event estimates for the Pegase 2869 model.
The columns are as in Table~\ref{tab:winterpeak}.
$^*$The se was $0$ because there were never two or more failures in any sample.
See text for discussion.
}
\end{table}

In addition to the examples above we investigated
IEEE case 14,
IEEE case 300, and
Pegase 1354, which were all dominated by one failure.
We considered a system which included random and correlated
wind power generators, but phase failure was not a rare event in
that system. Pegase 13659 was too large for our computation.
The Laplacian matrix has $37{,}250$ rows and columns and we use the SVD
to compute the generalized inverses we need.
Pegase 9241 was large enough to be very slow and it did
not have rare failures.

In our numerical tables we have used the plain sample variance 
of the importance sampled values to compute a standard 
error.  \cite{botev2015tail} note that the resulting 
standard error can be very inefficient and they propose 
a superior estimator. 
A naive implementation of their method would cost $O(J^4)$
for $J$ linear constraints but they are able to reduce 
that cost to $O(J^3)$. We have not used that method here 
because with $J=5772$ (Polish winter peak) or $J=7936$ (Pegase 2869), 
even $J^3$ is too much to pay for a better variance estimate.

\section{Discussion}\label{sec:discussion}

We have introduced a version of mixture importance sampling
for problems with multiple failure modes.  The sample values
are constrained to have at least one failure and we obtain
bounded relative error.

The ALOE importance sampler is more accurate than
a state of the art code for computing high dimensional Gaussian
probabilities in our rare event setting, but not otherwise.
It is also more reliable than the recent directional sampling
method of \cite{ahn:kim:2018}.  That method gains accuracy
by integrating over a randomly chosen line through the origin in $\real^d$,
but it samples those lines uniformly and not by importance sampling.
It is possible to combine the ideas, sampling $\bsx$ by ALOE
and then integrating over the line defined by the unit vector $\bsx/\Vert\bsx\Vert$.
Preliminary results show this to improve upon ALOE in the power grid problem
however a full discussion would add too much length and detail to the present paper.

We have noticed two areas where ALOE can be improved.
First, if we never see $S\ge2$ concurrent rare events, ALOE
will return the union bound $\hat\mu =\bar\mu$, with an
estimated variance of zero. That variance estimate could be
undesirable even when $\mu/\bar\mu\approx1$.  Because $S(\bsx)$ is supported
on $\{1,1/2,\dots,1/J\}$ we can get an interval estimate of $\mu$
by putting a multinomial prior on this support set and using the
posterior distribution given the sample. 
The solution in \cite{botev2015tail} is attractive when $J$
is not so large.

A second and related issue is that while $\bar\mu \ge\hat\mu\ge\bar\mu/J$
always occurs, it is possible to get $\hat\mu < \umu =\max_{1\le j\le J}P_j$.
We have seen this in cases where $\bar\mu\approx\umu$ because one
of the $P_j$ dominates all of the others combined. In such cases $\mu$, $\bar\mu$
and $\umu$ are all very close together and $\hat\mu$ has small relative standard
deviation.  Improving these two issues is outside the scope of this paper.  They are
both things that happen in cases where we already have a very good idea of the magnitude of $\mu$.
The problem was much less severe for ALOE than it was for directional sampling.
For instance, ALOE will not give $\hat\mu<\bar\mu/J$, while directional sampling can.

In large problems the
algebra can potentially be reduced by ignoring
the very rarest events and simply adding their probabilities
to the estimate.  This will provide a mildly conservative bound.
There is also the possibility of exploiting many generalized upper
and lower bounds on the probability of a union.
See for instance the survey by \cite{yang:alaj:taka:2014}.

\section*{Acknowledgments}

We thank Alan Genz for providing details of the {\tt mvtnorm} package.
We thank Yanbo Tang and Jeffrey Negrea for noticing that
Lemma~\ref{lem:prodmomentbound} could be proved
by Cauchy-Schwarz, which is shorter than our original proof
and also establishes necessity.
Thanks also to Bert Zwart, Jose Blanchet, David Siegmund, Mark Huber, Zdravko Botev and an
anonymous reviewer for pointers to the literature.
ABO thanks the Center for Nonlinear Studies at Los Alamos for their
hospitality while he visited.
%The work of YM and MC was funded by DOE/GMLC 2.0 project: ``Emergency Monitoring and controls through new technologies and analytics''. 
The work of YM and MC was supported by CNLS and  DOE/GMLC 2.0 project: ``Emergency Monitoring and controls through new technologies and analytics''. 
ABO was supported by the NSF under grants
DMS-1521145 and DMS-1407397.
\bibliographystyle{chicago}
\bibliography{rare,instanton}

%\section{Appendix: Proofs}\label{sec:proofs}
%\appendix{Appendix: Proofs}\label{sec:proofs}
\section*{Appendix: Proofs}

\subsection*{Analysis of self-normalized importance sampling}\label{sec:snisbad}

Here we show that self-normalized importance sampling cannot attain
the smallest variances in a rare event setting.  We build on a remark by
\citet[Chapter 2]{hest:1988} 
and follow a derivation from \citet[Chapter 9]{mcbook}.
The key problem is that
the optimal self-normalized importance sampler for a rare
event places only $1/2$ of its probability in the rare event. 

\citet[Chapter 2]{hest:1988}  quoting \cite{kahn:mars:1953}
notes that the optimal self-normalized importance sampling 
density $q$ for estimating $\mu=\e(f(\bsx))$ when $\bsx\sim p$ is proportional to $|f(\bsx)-\mu|p(\bsx)$. 
Taking $f(\bsx) = \one\{\bsx\in A\}$, we get $q(\bsx) = |\one\{\bsx\in A\}-\mu|p(\bsx)/c$
for some $c>0$. Solving $q(A\cup A^c)=1$ yields $c = 2\mu(1-\mu)$ and then $q(A)=q(A^c)=1/2$.
The self-normalized importance sampler is a ratio estimate $\sum_i f(\bsx_i)p(\bsx_i)/q(\bsx_i)\bigm/\sum_i p(\bsx_i)/q(\bsx_i)$.
Its variance is asymptotic to $\sigma^2_q/n$ for
\begin{align*}
\sigma^2_q=
\e_q\Bigl( \frac{p(\bsx)^2(f(\bsx)-\mu)^2}{q(\bsx)^2}\Bigr) 
\end{align*}
Half of the time $q$ places $\bsx\in A$ yielding
$p(\bsx)/q(\bsx)=(1-\mu)^{-1}/c$ and half of the time $\bsx\in A^c$
with $p(\bsx)/q(\bsx)=\mu^{-1}/c$. Thus
\begin{align*}
\sigma^2_q=
\frac12\Bigl[\frac1{c^2}+\frac1{c^2}\Bigr]=4\mu^2(1-\mu)^2.
\end{align*}
The best possible asymptotic coefficient of variation for SNIS is then
$2(1-\mu)/\sqrt{n}\approx 2/\sqrt{n}$ for rare events.
%
%=\frac12\Bigl[
%\frac{(1-\mu)^2}{c^2(1-\mu)^2}+\frac1{c^2\mu^2}
%\Bigr]

\subsection*{Proof of Theorem~\ref{thm:varub}}

Our motivating problem involves probabilities defined by Gaussian 
content of half-spaces. The approach generalizes to 
estimating the probability of the union of any finite set of events. 
We can also consider a countable number of events when 
the union bound is finite;
see remarks below.
We will assume that each event has positive probability, but 
that condition can also be weakened to a positive union bound,
as described at the end of this section.

For definiteness, we define our sets in terms of indicator 
functions of a random variable $\bsx\in\real^d$ with 
probability density $p$.  
The same formulas work for 
general sample spaces and the density can be with respect 
to an arbitrary base measure. 

We cast our notation into this more general setting as follows. 
For $J\ge1$, and $j=1,\dots,J$, let the  subset $H_j\subset\real^d$ define both 
the event $H_j = \one\{\bsx\in H_j\}$ and the 
indicator function $H_j(\bsx) = 1\{\bsx\in H_j\}$. 
For $u\subseteq\{1,\dots,J\}$
we let $H_u=\cup_{j\in u}H_j$ and 
$H_u(\bsx)=\max_{j\in u}H_j(\bsx)$,
with $H_\emptyset(\bsx)=0$. 
As before $P_u=\e(H_u(\bsx))$,
the number of events is $S(\bsx) =\sum_{j=1}^JH_j(\bsx)$
and $\Pr(S=s)=T_s$. 

Recall that we use $-u$ for complements with respect to $1{:}J$, especially 
within subscripts, and $H_u^c(\bsx)$
for the complementary outcome $1-H_u(\bsx)$. 
Then $H_u(\bsx)H_{-u}^c(\bsx)$ describes the 
event where $\bsx_j\in H_j$ if and only if $j\in u$. 

If $P_j>0$, then the distribution $q_j$ of $\bsx$ given $H_j$
is well defined: $q_j(\bsx) = p(\bsx)H_j(\bsx)/P_j$. 
If $\min_jP_j>0$ then we can define the mixture distribution 
\begin{align}
\label{eq:qalphageneral}
q_{\alpha^*} = \sum_{j=1}^J\alpha_j^*q_j,\quad \alpha^*_j = P_j/\bar \mu,\quad 
\bar\mu =\sum_{j=1}^JP_j. 
\end{align}
For $n\ge1$, our estimator of $\mu = \Pr( S(\bsx)>0)$ is 
\begin{align}\label{eq:hatmugeneral}
\hat\mu_{\alpha^*}=\frac{\bar\mu}n\sum_{i=1}^n\frac1{S(\bsx_i)},\quad 
\bsx_i\simiid q_{\alpha^*}. 
\end{align}

In this section, some equations include both 
randomness due to $\bsx_i\sim q_{\alpha^*}$
and randomness due to $\bsx\sim p$. For section only,
we use $\Pr_{*}$, $\e_{*}$ and $\var_{*}$ when 
the randomness is from observations $\bsx_i\sim q_{\alpha^*}$,
while $\e$, $\Pr$ and $\var$ are with respect to $\bsx\sim p$.

\begin{theorem}\label{thm:varub}
If $1\le J<\infty$ and $\min_jP_j>0$ and $n\ge1$,
then $\hat\mu_{\alpha^*}$ defined by~\eqref{eq:qalphageneral}
satisfies 
$\Pr_*({\bar\mu} /J \le \hat \mu_{\alpha^*}\le \bar\mu)=1$,
\begin{align}\label{eq:unbigen}
\e_*(\hat\mu_{\alpha^*})=\mu,
\end{align}
and 
\begin{align}\label{eq:varubgen}
\var_*(\hat\mu_{\alpha^*}) 
%=\frac1n\biggl( \bar\mu\sum_{|u|>0}|u|^{-1}Q_u-\mu^2\biggr) 
=\frac1n\biggl( \bar\mu\sum_{s=1}^J\frac{T_s}s 
%s^{-1}R_s 
-\mu^2\biggr) 
\le \frac{\mu(\bar\mu-\mu)}{n}. 
\end{align}
\end{theorem}

\begin{proof}[Proof of Theorem~\ref{thm:varub}.]
%Both conclusions hold trivially if $\bar\mu=0$, so we assume that $\bar\mu>0$. 
Let $H(\bsx) = \max_{1\le j\le J}H_j(\bsx)$
and $H=\{\bsx\mid H(\bsx)=1\}$. 
If $\bsx_i\sim q_{\alpha^*}$,
then $\bsx_i\in H$ always holds. 
Then $1\le S(\bsx_i)\le J$ holds establishing the bounds on 
$\hat\mu_{\alpha^*}$. 
Next 
\begin{align*}
\e_{*}\Biggl( \biggl(\,\sum_{j=1}^JH_j(\bsx_1)\biggr)^{-1}\Biggr) 
&=
\sum_{\ell=1}^J\frac{P_\ell}{\bar\mu}\int_{H}
\frac{H_\ell(\bsx)P_\ell^{-1}p(\bsx)}{\sum_{j=1}^JH_j(\bsx)}
\rd\bsx 
=\frac1{\bar\mu}\int_{H} p(\bsx)\rd\bsx = \frac{\mu}{\bar\mu},
\end{align*}
establishing~\eqref{eq:unbigen}. 

Because $\hat\mu_{\alpha^*}$ is unbiased its variance is 
\begin{align}\label{eq:itsvar}
\frac1 
n\Biggl( 
\bar\mu^2\e_{*} \Biggl( 
\Biggl(\frac{ H(\bsx_1)}{\sum_{j=1}^JH_j(\bsx_1)}\Biggr)^2 \Biggr) 
-\mu^2 
\Biggr). 
\end{align}
Next 
\begin{align*}
\e_{*}\Bigl( \Bigl(\frac{H(\bsx_1)}{\sum_jH_j(\bsx_1)}\Bigr)^2\Bigr) 
&=
\sum_{j=1}^J 
\frac{P_j}{\bar\mu}
\int_H 
\Bigl(\frac{H(\bsx)}{1+\sum_{\ell\ne j}H_\ell(\bsx)}\Bigr)^2 
\frac{p(\bsx)H_j(\bsx)}{P_j}
\rd\bsx\\
&=
\bar\mu^{-1}
\int_H \frac{p(\bsx)}{\sum_{j=1}^JH_j(\bsx)} \rd\bsx\\
&=
\bar\mu^{-1}\sum_{|u|>0}\frac1{|u|}
\int H_u(\bsx) H_{-u}^c(\bsx) p(\bsx)\rd\bsx\\
&=
\bar\mu^{-1}\sum_{s=1}^J\frac1{s}T_s,
%\\
%&=\bar\mu^{-1}\sum_{|u|>0}|u|^{-1}Q_u 
\end{align*}
which, with~\eqref{eq:itsvar},
establishes the equality in~\eqref{eq:varubgen}. 
Next 
\begin{align*}
\bar\mu^{-1}\sum_{s=1}^J\frac{T_s}s 
&
\le 
\bar\mu^{-1}\sum_{s=1}^JT_s 
=\bar\mu^{-1}(1-T_0) 
= \bar\mu^{-1}\mu. 
\end{align*}
Finally 
$\bar\mu^2(\bar\mu^{-1}\mu)-\mu^2 = (\bar\mu-\mu)\mu$,
establishing the upper bound in~\eqref{eq:varubgen}. 
\end{proof}

We can generalize the previous theorem to higher moments. 
Our estimate is an average of $\bar\mu/S(\bsx_i)$. 
The $k$'th moment of this quantity is 
\begin{align*}
\e_*\Bigl(\Bigl(
\frac{\bar\mu}{S(\bsx_1)}
\Bigr)^k\,\Bigr) 
&=\bar\mu^k\sum_{j=1}^J\alpha_j\int_HS(\bsx)^{-k}q_j(\bsx)\rd\bsx\\
&=\bar\mu^k\sum_{j=1}^J\frac{P_j}{\bar\mu}\int_HS(\bsx)^{-k}
\frac{p(\bsx)H_j(\bsx)}{P_j}\rd\bsx\\
&=\bar\mu^{k-1}\int_H S(\bsx)^{1-k}\rd\bsx\\
&=\sum_{s=1}^JT_s\Bigl(\frac{ \bar\mu}s\Bigr)^{k-1}. 
\end{align*}

\begin{remark}
Suppose that one of the $P_j=0$
but $\bar\mu>0$. In this case $q_j$ is not 
well defined. However $q_{\alpha^*}$ places 
probability $0$ on $q_j$, so we may delete 
the $q_j$ component without changing the algorithm 
and then sampling from $q_{\alpha^*}$
is well defined. 
%In an application we might well know that $P_j=0$
%and delete it explicitly. 
\end{remark}

\begin{remark}
Next suppose that there are infinitely many events,
one for each $j\in\natu$. 
If $\bar \mu\in(0,\infty)$, then $q_{\alpha^*}$ is well defined. 
The same proof goes through, only now sums over $1{:}J$
must be replaced by sums over $\natu$. 
\end{remark}

\subsection*{Proof of Lemma~\ref{lem:prodmomentbound}}

From the Cauchy-Schwarz inequality,
\begin{align}\label{eq:cs}
1-\e(S)\e(S^{-1})
=
\cov(S,S^{-1})
\ge 
-\sqrt{\var(S)\var(S^{-1})}.
\end{align}
Now $\var(S)\le (J-1)^2/4$ and
$\var(S^{-1})\le (1-J^{-1})^2/4$ 
because the support of $S$ is in $[1,J]$.
Therefore
$$
\e(S)\e(S^{-1}) \le 1 + \frac{(J-1)(1-J^{-1})}4=
\frac{ J + J^{-1} +2}4.
$$
Finally, the unique distribution
for which $\var(S)$, $\var(S^{-1})$ and $-\cor(S,S^{-1})$
all attain their maxima is $\dustd\{1,J\}$.
$\qed$
\subsubsection*{Generalization}

The lemma generalizes.
If $\Pr(a\le X\le b)=1$ for 
$0<a\le b<\infty$ then the same argument
yields $\e(X)\e(X^{-1})\le (a/b+b/a+2)/4$.

\end{document}